\documentclass[english]{cccconf}
\usepackage[comma,numbers,square,sort&compress]{natbib}
\usepackage{amsthm}
\usepackage{newtxtext}      
\usepackage{newtxmath}      

\usepackage{graphicx, graphics, epsfig}   
\usepackage{booktabs, threeparttable}     
\usepackage{array, multirow}          

\usepackage{amsmath,  amsfonts} 
  
\usepackage{amssymb}
\usepackage{mathtools}                  
\usepackage{bm}                          
\usepackage{mathrsfs}                    
\usepackage{dsfont}                      
\usepackage{subeqnarray}

\usepackage{algorithm}
\usepackage{algorithmicx}
\usepackage[noend]{algpseudocode}

\usepackage{xcolor}                       
\usepackage{enumerate,enumitem,subcaption}                 

\usepackage{natbib}                       

\usepackage{braket}                      
\usepackage{times}                       
\usepackage{caption}
\usepackage{comment}

\algdef{SE}[DOWHILE]{Do}{doWhile}{\algorithmicdo}[1]{\algorithmicwhile\ #1}
\newtheorem{theorem}{Theorem}
\newtheorem{proposition}[theorem]{Proposition}
\newtheorem{definition}[theorem]{Definition}
\newtheorem{assumption}[theorem]{Assumption}
\newtheorem{lemma}[theorem]{Lemma}

\newtheorem{remark}{Remark}[section]
\usepackage{epstopdf}

\allowdisplaybreaks[4] 

\makeatletter
\newcommand{\Rmnum}[1]{\expandafter\@slowromancap\romannumeral #1@}
\makeatother
\begin{document}

\title{A Characterization of Integral Input-to-state  Stability for Hybrid Systems with Memory}

\author{Wenbang Wang\aref{amss},
        Neng Li\aref{amss},
        Wei Ren\aref{amss}}



\affiliation[amss]{School of Control Science and Engineering,
       Dalian University of Technology,  Dalian 116024, P.~R.~China
        \email{wang2002@mail.dlut.edu.cn;744845870@mail.dlut.edu.cn;wei.ren@dlut.edu.cn}}

\maketitle

\begin{abstract}
This paper addresses characterizations of Integral Input-to-State Stability (iISS) for hybrid systems with memory. 
Based on the Krasovskii approach, a novel Lyapunov characterization of iISS is established to extend the hybrid system theory to the time-delay case. 
In particular, we introduce the notions of dissipativity, detectability and storage functional to describe the iISS property from different perspectives. 
Under mild regularity and convexity assumptions, the equivalence relations among diverse stability descriptions are established, which lays a solid foundation for the control design.  
Finally, a numerical example is presented to illustrate the derived results.
\end{abstract}


\footnotetext{This work was supported by the Fundamental Research Funds for the Central Universities under Grant DUT22RT(3)090, the LiaoNing Revitalization Talents Program (XLYC2403048), and the National Natural Science Foundation of China under Grant 62573085.}

\section{Introduction}
	Hybrid systems serve as a powerful framework to integrate continuous flows described by differential equations with discrete jumps described by difference equations or logical rules. Because of the strong capability, hybrid systems play an essential role in the modeling and analysis of modern complex engineering systems \cite{HybridDynamicalSystemsModeling}, such as networked control systems \cite{zhang2018guaranteed}, robotic motion control \cite{westervelt2003hybrid}, and mechanical systems \cite{marton2009control}. For instance, networked control systems are often modeled as either switched systems or impulsive systems,  which are two special types of hybrid systems. In practical applications, external disturbances are inevitable and have effects on dynamical systems. One of the main effects is time delays, which may result from remote communication, sensor or actuator latencies, and sample-and-hold implementations \cite{Kamal2021}. Typical examples include switched time-delay systems \cite{mahmoud2010switched} and impulsive time-delay systems \cite{li2021stability}. In particular, hybrid systems with time delays are called hybrid systems with memory.

Due to the existence of time delays, it is not easy to consider the stability analysis of hybrid systems with memory \cite{Hajdu2017}. Nevertheless, extensive research has been conducted to overcome these difficulties. Specifically, a general framework for Lyapunov-based sufficient stability conditions for hybrid systems with memory was established in \cite{LiuTeel}, which was subsequently extended to investigate Input-to-State Stability (ISS) in \cite{renwei}. For switched time-delay systems, the Average Dwell Time (ADT) technique has been widely adopted to mitigate the destabilizing effects of switching signals \cite{yan2008stability}. The work in \cite{Sun2006} reduced conservatism by employing  linear matrix inequalities. Regarding impulsive time-delay systems, the Razumikhin method and comparison principle were employed in \cite{liu2001uniform} to address the jump-delay coupling. Furthermore, these results were generalized to hybrid systems with memory in \cite{LiuSun} via Razumikhin-type theorems.

Existing works mainly focus on the classic stability properties, including asymptotic stability and ISS. However, Global Asymptotic Stability (GAS) fails to reflect the effects of noise and uncertainty commonly present in real-world operating conditions. ISS  requires systems to maintain bounded states for any input with bounded magnitude \cite{Pepetimedelay}. This stringent robustness requirement is often overly demanding for systems with time delays, frequently leading to conservative delay-independent stability conditions. In contrast, Integral Input-to-State Stability (iISS) \cite{Jiang2004Unifying}  adopts an energy dissipation perspective, allowing for input magnitudes to be arbitrarily large or even unbounded for short durations. System stability is guaranteed as long as the integral energy remains finite \cite{Liu2020ISS}. For hybrid systems with memory, conducting iISS analysis not only enables more effective handling of sudden, intense disturbances or non-persistent attacks during pulse instants but is also critical for energy-constrained robust control design \cite{iissNCS}.

Despite the theoretical appeal of iISS, established stability tools for subsystems often fail to generalize to the unified framework of hybrid systems with memory. Existing methods for purely switched \cite{Liu2022iISSADT} or impulsive systems \cite{ Chen2009ImpulsiveDelay} struggle to handle the complex coupling between history-dependent dynamics and discrete jumps. The primary challenge lies in the construction of LKFs \cite{Chaillet2022LKFiISS}, which accounts for the energy stored in the history segment and the non-monotonic hybrid behaviors simultaneously. Specifically, the integral terms in LKFs do not reset instantaneously during impulses, creating a mathematical mismatch between the discrete state updates and the continuous functional energy evolution. Consequently, a coherent and constructive framework for establishing iISS criteria for hybrid systems with memory remains lacking.

To address this gap, this paper develops a unified theoretical framework for iISS of hybrid systems with memory. The main contributions of this paper are summarized as follows. First, we extend the notions of iISS and dissipativity to the framework of hybrid systems with memory. Second, we construct verifiable criteria using LKFs and storage functionals tailored for history-dependent dynamics and jump behaviors. Third, we establish the equivalence among LKF-based iISS criteria, storage functionals, and asymptotic gain estimates under mild regularity assumptions. The resulting theorems clarify how integral input effects, history dependence, and discrete transitions collectively determine stability characteristics. 

The rest of this paper is organized below. Section 2 presents the fundamentals and proposes stability properties of hybrid systems with memory. Section 3 presents the main results. A numerical example is given in Section 4 to illustrate the derived results, followed by the conclusion in Section 5.

\section{Preliminaries}

	Let $\mathbb{R}:=(-\infty,+\infty)$, $\mathbb{R}_{\geq0}:=[0,+\infty)$, $\mathbb{R}_{\leq0}:=(-\infty,0]$, $\mathbb{Z}_{\geq0}:=\{0,1,\ldots\}$ and $\mathbb{Z}_{\leq0}:=\{0,-1,\ldots\}$. Let $\mathbb{R}^n$ be the $n$-dimensional Euclidean space, $\mathbb{B}$ be the open unit ball in $\mathbb{R}^n$, and $\varepsilon\bar{\mathbb{B}}$ be the closed ball of radius $\varepsilon>0$. For a vector or matrix $A$, $A^\top$ is its transpose and $|A|$ is the Euclidean norm. For $(t,j),(s,k)\in\mathbb{R}^2$,  $(t,j)\preceq(s,k)$ holds if $t+j\leq s+k$ and $(t,j)\prec(s,k)$ holds if $t+j<s+k$. For a set $\mathcal{A}\subset\mathbb{R}^n$, its closure and closed convex hull are denoted as $\bar{\mathcal{A}}$ and $\overline{co}\mathcal{A}$, respectively.The distance from a point $x\in\mathbb{R}^n$ to a set $\mathcal{A}\subset\mathbb{R}^n$ is defined as $|x|_\mathcal{A}:=\inf_{y\in\mathcal{A}}|x-y|$. For $\mathcal{A},\mathcal{B}\subset\mathbb{R}^n$ with $\mathcal{A}\subset\mathcal{B}$, the set $\mathcal{A}$ is said to be relatively closed in $\mathcal{B}$ if $\mathcal{A}=\bar{\mathcal{A}}\cap\mathcal{B}$. 
A continuous function $\alpha:\mathbb{R}_{\geq0}\to\mathbb{R}_{\geq0}$ is of class $\mathcal{K}$, if it is positive definite and strictly increasing. A function $\alpha$ is said to be of class $\mathcal{K}_{\infty}$ if it is of class $\mathcal{K}$ and, in addition, is unbounded. A function $\beta: \mathbb{R}_{\geq0}\times\mathbb{R}_{\geq0}\to\mathbb{R}_{\geq0}$ is said to be of class $\mathcal{KL}$ if for every fixed $t\geq0$, the function $\beta(\cdot,t)$ belongs to class $\mathcal{K}$, and for every fixed $s\geq0$, $\beta(s,t)$ decreases to zero as $t\to\infty$. A function $\beta:\mathbb{R}_{\geq0}\times\mathbb{R}_{\geq0}\times\mathbb{R}_{\geq0}\to\mathbb{R}_{\geq0}$ is said to be of class $\mathcal{KLL}$ if for every fixed $t\geq0$, the function $(r,s)\mapsto\beta(r,s,t)$ belongs  to class $\mathcal{KL}$, and for every fixed $s\geq0$, the function $(r,t)\mapsto\beta(r,s,t)$  belongs  to class $\mathcal{KL}$. A set-valued mapping $\mathcal{F}:\mathcal{O}\rightrightarrows\mathbb{R}^n$ is outer semicontinuous at $x\in\mathcal{O}$, if for every sequence $x_i\to x$ and every sequence $y_i\to y$ such that $y_i\in\mathcal{F}(x_i)$ for all $i$, it follows that $y\in\mathcal{F}(x)$. A set-valued mapping $\mathcal{F}:\mathcal{O}\rightrightarrows\mathbb{R}^n$ is locally bounded, if for every compact set $K\subset\mathcal{O}$, there exists a compact set $K^{\prime}\subset\mathbb{R}^n$ such that $\mathcal{F}(K)\subset K^\prime$. For functions $\alpha,\beta,h:\mathbb{R}_{\ge0}\to\mathbb{R}_{\ge0}$ and a scalar $c>0$, the composition $(\alpha\circ\beta\circ h)(\nu):=\alpha(\beta(h(\nu)))$ is used for all $\nu\ge0$.
\subsection{Hybrid Systems with Memory}
In this subsection, the fundamental concepts of hybrid systems with memory are recalled from \cite{LiuTeel}.

A set $E\subseteq\mathbb{R}\times\mathbb{Z}$ is called a \emph{compact hybrid time domain with memory} if $E_{\geq0}=\cup_{j=0}^{J-1}([t_j,t_{j+1}],j)$ and $E_{\leq0}=\cup_{k=0}^K([s_k,s_{k-1}],-k+1)$ with $s_K\leq\cdots\leq s_0=0=t_0\leq\cdots\leq t_J$. The set $E$ is called a \emph{hybrid time domain with memory} if for all $(T,J)\in E_{\geq0}$ and $(S,K)\in\mathbb{R}_{\geq0}\times\mathbb{Z}_{\geq0}$, the set $(E_{\geq0}\cap([0,T]\times\{0,\ldots,J\}))\cup(E_{\leq0}\cap([-S,0]\times\{-K,\ldots,0\}))$ is a compact hybrid time domain with memory. Especially, $E_{\leq0}$ is called a \emph{hybrid memory domain}. A hybrid signal is a function defined on a hybrid time domain with memory. Let $\operatorname{dom}_{\geq0}x:=\operatorname{dom}x\cap(\mathbb{R}_{\geq0}\times\mathbb{Z}_{\geq0})$ and $\operatorname{dom}_{\leq0}(x):=$ $\operatorname{dom}x\cap(\mathbb{R}_{\leq0}\times\mathbb{Z}_{\leq0})$. A hybrid signal $u: \operatorname{dom}u\to\mathcal{U}$ is called a \emph{hybrid input}, if it satisfies  $\operatorname{dom}_{\leq0}(u)\equiv\{(0,0)\}$, and for each $j\in \mathbb{Z}_{\geq0}$, the function $u(\cdot,j)$ is Lebesgue measurable and locally essentially bounded on $I_u^j=\{t: (t,j)\in\operatorname{dom}u\}$. A hybrid signal $x:\operatorname{dom}x\to\mathbb{R}^n$ is called a \emph{hybrid arc with memory}, if $x(\cdot,j)$ is locally absolutely continuous on $I_x^j=\{t: (t,j)\in\operatorname{dom}x\}$ with $j\in\mathbb{Z}$. A \emph{hybrid memory arc} is a hybrid arc with memory whose time domain is a hybrid memory domain. 

Given a hybrid time domain with memory $E$, define $\sup_{t\geq 0} E:=\sup\{t\in\mathbb{R}_{\geq0}: \exists j\in\mathbb{Z}_{\geq 0} \text{ s.t. } (t, j) \in E_{\geq 0}\}$ and $\inf_{t\leq 0}E:=\inf\{t\in\mathbb{R}_{\leq0}: \exists j \in \mathbb{Z}_{\leq 0} \text{ s.t. } (t, j) \in E_{\leq 0} \}$. Similarly, $\sup_{j\geq0}E$ and $\inf_{j\leq0} E$ can be defined. Denote $\mathcal{L}_{\geq0}(E):=\sup_{t\geq0}E+\sup_{j\geq0}E$ and $\mathcal{L}_{\leq0}(E):=\inf_{t\leq0}E+\inf_{j\leq0}E$. Given $\Delta\in\mathbb{R}_{\geq0}$, $\mathcal{M}^\Delta$ denotes the set of the hybrid memory arcs $\varphi$ satisfying $-\Delta-1\leq\mathcal{L}_{\leq0}(\operatorname{dom}\varphi)\leq$ $-\Delta$. Given a hybrid arc with memory $x$, the operator $\mathcal{A}_{[\cdot,\cdot]}^{\Delta}x$ : $\operatorname{dom}_{\geq0}x\to\mathcal{M}^\Delta$ is defined by $\mathcal{A}_{[t,j]}^\Delta x(s,k)=x(t+s,j+k)$ for all $(s,k)\in\operatorname{dom}\mathcal{A}_{[t,j]}^\Delta x$, where $(t,j)\in\operatorname{dom}_{\geq0}x$, $\operatorname{dom}\mathcal{A}_{[t,j]}^\Delta x:=\{(s,k)\in\mathbb{R}_{\leq0}\times\mathbb{Z}_{\leq0}: (t+s,j+k)\in\operatorname{dom}x,s+k\geq-\Delta_{\mathrm{inf}}\}$, and $\Delta_{\mathrm{inf}}:=\inf\{\delta\geq\Delta: \exists(t+s,j+k)\in\operatorname{dom}x \text{ such that } s+k=-\delta\}$. Let $\mathcal{A}^\Delta x:=\mathcal{A}_{[\cdot,\cdot]}^\Delta x$ if the time argument can be omitted.

Consider the hybrid system  $\mathcal{H}_\mathcal{M}^\Delta:=(\mathcal{F},\mathcal{G},\mathcal{C},\mathcal{D},\mathbb{R}^n,\mathcal{U})$:
\begin{align} 
	\label{eq:hybrid}
	\left\{\begin{aligned}
		&\dot{x}\in\mathcal{F}(\mathcal{A}^\Delta x,u)  &\quad&\forall(\mathcal{A}^\Delta x,u)\in\mathcal{C} ,\\ 
		&x^+\in\mathcal{G}(\mathcal{A}^\Delta x,u)  &\quad&\forall(\mathcal{A}^\Delta x,u)\in\mathcal{D} , 
	\end{aligned}\right.
\end{align}
where $x\in\mathbb{R}^n$ is the state, $u\in\mathcal{U}\subset\mathbb{R}^m$ is the input, $\mathcal{C}\subseteq\mathcal{M}^{\Delta}\times\mathcal{U}$ is the flow set, and $\mathcal{D}\subseteq\mathcal{M}^{\Delta}\times\mathcal{U}$ is the jump set.  
\begin{assumption}\cite{renwei}
	The following conditions  guarantee the existence of solutions for system \eqref{eq:hybrid}.
	\label{assume}
	\begin{enumerate}
		\item $\mathcal{U}\subset\mathbb{R}^m$ is a closed set; $\mathcal{C}\cap(\mathcal{M}^\Delta\times\mathcal{U})$ and $\mathcal{D}\cap(\mathcal{M}^\Delta\times\mathcal{U})$ are
		relatively closed in $\mathcal{M}^{\Delta}\times\mathcal{U};$
		\item The mapping $\mathcal{F}:\mathcal{C}\rightrightarrows\mathbb{R}^n$ is outer semicontinuous and locally bounded relative to $\mathcal{C}$, with $\mathcal{F}(\varphi,u)$ being nonempty and convex for each $(\varphi,u)\in\mathcal{C};$
		\item The mapping $\mathcal{G}:\mathcal{D}\rightrightarrows\mathbb{R}^n$ is outer semicontinuous and locally bounded relative to $\mathcal{D}$, with $\mathcal{G}(\varphi,u)$ being nonempty  for each $(\varphi,u)\in\mathcal{D}.$
	\end{enumerate}
\end{assumption}

\begin{definition}
	A hybrid arc with memory $x$ and a hybrid input $u$ are
	called a \emph{solution pair} $(x,u)$ to $\mathcal{H}_\mathcal{M}^\Delta$, if
	\begin{enumerate}
		\item $\operatorname{dom}_{\geq 0}\,x = \operatorname{dom}\,u$ \ and \ 
		$(\mathcal{A}_{[0,0]}^\Delta x, u(0,0)) \in \mathcal{C} \cup \mathcal{D};$
		
		\item for all $j \in \mathbb{Z}_{\geq 0}$ and almost all $t$ with 
		$(t,j) \in \operatorname{dom}_{\geq 0}\,x$,
		$(\mathcal{A}_{[t,j]}^\Delta x, u(t,j)) \in \mathcal{C}$ \ and \
		$\dot{x}(t,j) \in \mathcal{F}(\mathcal{A}_{[t,j]}^\Delta x, u(t,j));$
		
		\item for all $(t,j) \in \operatorname{dom}_{\geq 0}\,x$ such that 
		$(t,j+1) \in \operatorname{dom}_{\geq 0}\,x$,
		$(\mathcal{A}_{[t,j]}^\Delta x, u(t,j)) \in \mathcal{D}$  and 
		$x(t,j+1) \in \mathcal{G}(\mathcal{A}_{[t,j]}^\Delta x,\\u(t,j)).$
	\end{enumerate}
\end{definition}

	Given any hybrid signal $z$, let $\sharp z:=\sup_{(t,j)\in\operatorname{dom}z}t+j$. For any hybrid input $u:\operatorname{dom}u\to\mathcal{U}$, let $(t_1,j_1),(t_2,j_2)\in\operatorname{dom}u$ and $(t_1,j_1)\preceq(t_2,j_2)$, and define
\begin{align*}
	\|u\|_{[(t_1,j_1),(t_2,j_2)]} &= \max \{\underset{(t,j)\in\Gamma(u), (t_1,j_1) \preceq (t,j) \preceq (t_2,j_2)}{\sup} |u(t,j)|, \\
	&\!\!\!\! \quad 
	\underset{(t,j)\in\operatorname{dom}u\setminus\Gamma(u), (t_1,j_1) \preceq (t,j) \preceq (t_2,j_2)}{\text{ess.}\sup} |u(t,j)|
	\},
\end{align*}
where $\Gamma(u):=\{(t,j)\in\operatorname{dom}u: (t,j+1)\in\operatorname{dom}u\}$. Define $\|u\|_{(t_2, j_2)}:=\|u\|_{[(0,0),(t_2,j_2)]}$ and $\|u\|_\sharp:=\lim_{t_2+j_2\to\sharp u}\|u\|_{(t_2,j_2)}$. Thus, $\|u\|_\infty:=\|u\|_\sharp$ if $\sharp u=+\infty$. 

A solution pair $(x,u)$ to $\mathcal{H}_\mathcal{M}^\Delta$ is \emph{complete} if $\operatorname{dom}_{\geq0}x$ is unbounded,  \emph{maximal} if it cannot be extended, and \emph{bounded} if there exists a compact set $\mathcal{O}\subset\mathbb{R}^n$ such that $x(t,j)\in\mathcal{O}$ for all $(t,j)\in\operatorname{dom}x$. $\mathfrak{S}_u^{\Delta}(\varphi)$ denotes the set of all the maximal solution pairs $(x,u)$ to $\mathcal{H}_\mathcal{M}^\Delta$ with the initial condition $\varphi\in\mathcal{M}^\Delta$ and finite $\|u\|_\sharp$. Assume that the system $\mathcal{H}_\mathcal{M}^\mathrm{\Delta}$ is \emph{forward complete}: for all $\varphi\in\mathcal{M}^\Delta$ and all inputs $u$ with finite $\left\|u\right\|_{\#}$, every solution pair $(x,u)\in\mathfrak{S}_u^\Delta(\varphi)$ is complete. For any $\varphi\in\mathcal{M}^\Delta$, let $\varphi_0=\mathcal{A}_{[0,0]}^\Delta x$, and $x(t, j, \varphi_0, u)$ is the state at $(t,j)\in\operatorname{dom}x$ starting from $\varphi_0$ under the control input $u\in\mathcal{U}$. For a set $\mathcal{W}\subset\mathbb{R}^n$ and $r\geq0$, define $\|\varphi\|_{\mathcal{W}}:=\sup_{s+k\in[-\Delta-1, 0]}|\varphi(s,k)|_{\mathcal{W}}$, and ${\mathcal{M}^\Delta}_r^{\mathcal{W}}:=\{\varphi\in\mathcal{M}^\Delta: \|\varphi\|_\mathcal{W}\leq r\}$.

For the system $\mathcal{H}_{\mathcal{M}}^\Delta$ and a set $B\subset\mathbb{R}^n$, a set $N\subset B$ is said to be \emph{weakly invariant}, if for every hybrid memory arc $\varphi$ with $\varphi(0,0)\in N$, there exists a complete solution pair $(x,u)\in\mathfrak{S}_u^\Delta(\varphi)$ such that $x(t,j) \in N$ for all $(t,j) \in \operatorname{dom}x$. Among all weakly invariant sets contained in $B$, the one that is largest with respect to set inclusion is called the \emph{largest weakly invariant set}, denoted by $\mathcal{N}_{\max}(B)$.

For a functional $V:\mathcal{M}^\Delta\to\mathbb{R}_{\geq0}$, the upper Dini derivative of $V$ at $\varphi\in\mathcal{M}^\Delta$ along the solutions of $\mathcal{H}_{\mathcal{M}}^{\Delta}$ is given by
\begin{align*}
	D^+ V(\varphi)&:= \sup_{(x,u)\in\mathfrak{S}_u^\Delta(\varphi)}\limsup_{h\to 0^+}\frac{V(\mathcal{A}_{[h,0]}^\Delta x)-V(\varphi)}{h}\\
	&:=\sup_{(x,u)\in\mathfrak{S}_u^\Delta(\varphi)} \Braket{\nabla V(\varphi),f(\varphi,u)}.
\end{align*}
\subsection{Stability Notions}
Consider the system \eqref{eq:hybrid}. The set $\mathcal{W}$ is \emph{0-input pre-stable}, if for each $\varepsilon >0$ there exists $\delta>0$ such that each solution pair $(x,0)\in\mathfrak{S}_u^{\Delta}(\varphi)$ with $\|\mathcal{A}_{[0,0]}^\Delta x\|_{\mathcal{W}}\leq\delta$ satisfies $|x(t,j)|_{\mathcal{W}}\leq\varepsilon$ for all $(t,j)\in \operatorname{dom}_{\geq0}x$.
The set $\mathcal{W}$ is said to be \emph{0-input pre-attractive}, if there exists $\delta>0$ such that each solution pair $(x,0)\in\mathfrak{S}_u^{\Delta}(\varphi)$ with  $\|\mathcal{A}_{[0,0]}^\Delta x\|_{\mathcal{W}}\leq\delta$ is bounded and, if complete, satisfies $\lim_{(t,j)\in \operatorname{dom}_{\geq0}x,t+j\to+\infty}|x(t,j,\varphi_0,0)|_{\mathcal{W}}=0$.
The set $\mathcal{W}$ is \emph{0-input pre-asymptotically stable}, if it is both 0-input pre-stable and 0-input pre-attractive.
\begin{definition} \label{defn iISS} 
	Consider the system \eqref{eq:hybrid}, a closed set $\mathcal{W}\subset\mathbb{R}^n$ is \emph{integral input-to-state stable (iISS)}, if there exist $\beta\in\mathcal{KLL}$, $\rho\in\mathcal{K}$ such that for all $(t, j)\in\operatorname{dom}_{\geq0}x$, and all $\mathcal{A}_{[0,0]}^\Delta x\in\mathcal{M}^\Delta$, each solution pair $(x,u)\in\mathfrak{S}_u^{\Delta}(\varphi)$ satisfies 
\begin{equation}
	\label{equation iISS}
	\hspace{-8pt}\scalebox{0.88}{\(\displaystyle
		|x(t,j)|_{\mathcal{W}}\leq\max\left\{\beta(\|\mathcal{A}_{[0,0]}^\Delta x\|_{\mathcal{W}},t,j),\\
		\sum_{n=0}^{j}\int_{t_{n}}^{t_{n+1}}\rho(|u|)\,ds\right\}.
		\)}
\end{equation}
The set $\mathcal{W}$ is said to be \emph{locally integral input-to-state stable (locally iISS)}, if there exist  $r_1,r_2>0$ such that, whenever $\|\mathcal{A}_{[0,0]}^{\Delta}x\|_{\mathcal{W}}\le r_1$ and a solution pair $(x,u)\in\mathfrak{S}_u^{\Delta}(\varphi)$ satisfies $\|u\|_{\infty}\le r_2$, the inequality \eqref{equation iISS} holds.
\end{definition}
\begin{definition}\label{defn o-input detectable} 
Let $\mathcal{W}\subset\mathbb{R}^n$ be a closed set satisfying $\mathcal{W}\subset N$. The set $\mathcal{W}$ is said to be \emph{0-input detectable} on $N$ for $\mathcal{H}_{\mathcal{M}}^\Delta$, if for every complete solution pair $(x,0)$ of $\mathcal{H}_{\mathcal{M}}^\Delta$ satisfying $x(t, j)\in N$ for all $(t,j)\in\operatorname{dom}_{\ge0}x$, we have
\begin{align*}
	\lim_{\substack{(t,j)\to+\infty, (t,j)\in\operatorname{dom}_{\ge0}x}}|x(t,j)|_{\mathcal{W}}=0.
\end{align*}
\end{definition} 
\begin{definition} \label{defn iISS Ly}
	Given a closed set $\mathcal{W}\subset\mathbb{R}^n$, a smooth functional $V:\mathcal{M}^\Delta\to\mathbb{R}_{\geq0}$ is called an \emph{iISS Lyapunov-Krasovskii functional (iISS-LKF)}   for \eqref{eq:hybrid}, if there exist 
$\alpha_1,\alpha_2,\alpha_3\in\mathcal{K}_\infty$ and $\rho\in\mathcal{K}$ such that
\begin{enumerate}
	\item for all $\varphi\in\mathcal{M}^\Delta,\alpha_1(|\varphi(0,0)|_{\mathcal{W}})\leq V(\varphi)\leq\alpha_2(\|\varphi\|_{\mathcal{W}})$;
	\item for all $(\varphi,u)\in\mathcal{C},D^+V(\varphi) \leq -\alpha_3(|\varphi(0,0)|_{\mathcal{W}}) + \rho(|u|)$;
	\item for all $(\varphi,u)\in\mathcal{D},V(\varphi^+)\leq V(\varphi)$.
\end{enumerate}	
\end{definition}
\begin{definition} 	\label{storge }
	Given a closed set $\mathcal{W}\subset\mathbb{R}^n$, a smooth functional $V:\mathcal{M}^\Delta\to\mathbb{R}_{\geq0}$ is called a \emph{storage functional}  for \eqref{eq:hybrid}, if there exist $\alpha_1,\alpha_2\in\mathcal{K}_\infty,\hat{\rho}\in\mathcal{K}$, and a continuous functional $\psi:\mathcal{M}^\Delta\to\mathbb{R}_{\geq0}$ such that
	\begin{align}  
		\alpha_1(|\varphi(0,0)|_\mathcal{W}) &\leq V(\varphi) \leq \alpha_2(\|\varphi\|_{\mathcal{W}}),  \!\!\!\!&&\forall \varphi\in\mathcal{M}^\Delta, \label{eq:5} \\
		D^+V(\varphi) &\leq -\psi(\varphi) +\hat{\rho}(|u|), \quad &&\forall(\varphi,u)\in\mathcal{C},\label{eq:6} \\
		V(\varphi^+) &\leq V(\varphi), \quad &&\forall(\varphi,u)\in\mathcal{D}. \label{eq:7} 
	\end{align}
\end{definition}
\begin{definition} \label{defn exp ly}
		Given a closed set $\mathcal{W}\subset\mathbb{R}^n$, an iISS-LKF $V:\mathcal{M}^\Delta\to\mathbb{R}_{\geq0}$ is said to have an \emph{exponential decay} for \eqref{eq:hybrid}, if there exist $\alpha_1, \alpha_2\in\mathcal{K}_\infty$, $\rho\in\mathcal{K}$, and $v\in(0,1]$ such that
	\begin{enumerate}
		\item for all $\varphi\in\mathcal{M}^\Delta,\alpha_1(|\varphi(0,0)|_\mathcal{W})\leq V(\varphi)\leq\alpha_2(\|\varphi\|_{\mathcal{W}})$;
		\item for all $(\varphi, u)\in\mathcal{C}$, $D^+V(\varphi)\leq-vV(\varphi)+\rho(|u|)$;
		\item for all $(\varphi, u)\in\mathcal{D}$, $V(\varphi^+)\leq V(\varphi)$.
	\end{enumerate}	
\end{definition}
\begin{definition} \label{defn asy  }
	The system \eqref{eq:hybrid} has the \emph{asymptotic gain property}, if there exists $\gamma\in\mathcal{K}$, such that each solution pair $(x,u)\in\mathfrak{S}_u^{\Delta}(\varphi)$ is bounded and if complete then
\begin{equation} 
	\label{asy} 	
	\underset{(t,j)\in\operatorname{dom}x, t+j\to\infty}{\lim\sup}|x(t,j)|_{\mathcal{W}}\leq\sum_{n=0}^{j}\int_{t_{n}}^{t_{n+1}}\gamma(|u|)ds.
\end{equation}
\end{definition}
\begin{definition} \label{defn gps }
	The system \eqref{eq:hybrid} is \emph{globally pre-stable}, if there exist $\alpha,\gamma \in \mathcal{K}$ such that
\begin{equation} \label{gps} 	
	\scalebox{0.9}{$
	|x(t,j)|_{\mathcal{W}}\leq\max\{\alpha(\|\mathcal{A}_{[0,0]}^\Delta x\|_{\mathcal{W}}),\sum_{n=0}^{j}\int_{t_{n}}^{t_{n+1}}\gamma(|u|)ds\}.
	$}
\end{equation}
\end{definition}
\begin{definition} \label{defn noniiss}
	The system \eqref{eq:hybrid} is \emph{nonuniformly iISS (non-iISS)} if it is globally pre-stable and has the asymptotic gain property.
\end{definition}
\begin{definition}
	The system \eqref{eq:hybrid} is \emph{smoothly dissipative} if there exist a storage functional $V:\mathcal{M}^\Delta\to\mathbb{R}_{\ge0}$ and a functional $\rho_d:\mathcal{M}^\Delta\to\mathbb{R}$ such that
\begin{enumerate}
	\item for every $\varphi\in\mathcal{C}$, $D^+V(\varphi)\le\rho_d(\varphi)$,
	\item for every $\varphi\in\mathcal{D}$, $V(\varphi^+)-V(\varphi)\le\rho_d(\varphi)$,
\end{enumerate}
where $\rho_d$ is called the dissipation rate functional.
\end{definition}

\section{Main Results}
This section presents the main results of this work. 
\begin{theorem} \label{mainresult}
	Let $\mathcal{W} \subset \mathbb{R}^n$ be a closed set, and let Assumption \ref{assume} hold. For each  $\varphi\in{\mathcal{M}}^{\Delta}$ and  $\varepsilon\geq0$, the set $\{\mathcal{F}(\varphi,u):u\in\mathcal{U}\cup \varepsilon\bar{\mathbb{B}}\}$ is convex. The following statements are equivalent.
\begin{enumerate}[label=\Roman*.]
	\item $\mathcal{H}_{\mathcal{M}}^{\Delta}$ admits an iISS-LKF with an exponential decay.
	\item $\mathcal{H}_{\mathcal{M}}^{\Delta}$ admits an iISS-LKF.
	\item $\mathcal{H}_{\mathcal{M}}^{\Delta}$ is iISS.
	\item $\mathcal{H}_{\mathcal{M}}^{\Delta}$ is nonuniformly iISS.
	\item $\mathcal{H}_{\mathcal{M}}^{\Delta}$ has the asymptotic gain property with respect to $\mathcal{W}$, and is locally iISS.
	\item$\mathcal{H}_{\mathcal{M}}^{\Delta}$ has the asymptotic gain property with respect to $\mathcal{W}$, and the set $\mathcal{W}$ is 0-input pre-stable for $\mathcal{H}_{\mathcal{M}}^{\Delta}$.
	\item $\mathcal{H}_{\mathcal{M}}^{\Delta}$ is smoothly dissipative with respect to $\mathcal{W}$ and the distance to $\mathcal{W}$ is 0-input detectable relative to $\{\xi \in \mathcal{M}^{\Delta} : \psi(\xi) = 0\}$ with $\psi $ given in \eqref{eq:6}.
	\item $\mathcal{H}_{\mathcal{M}}^{\Delta}$ is both 0-input pre-asymptotically stable and smoothly dissipative with respect to $\mathcal{W}$.
\end{enumerate}
\end{theorem} 
\begin{proof}
The implication \Rmnum{1} $\Rightarrow$ \Rmnum{2} follows directly from the lower bound property $V(\varphi)\ge\alpha_1(|\varphi(0,0)|_{\mathcal{W}})$. Combining this with the exponential decay condition yields $D^+V(\varphi) \leq -v\alpha_1(|\varphi(0,0)|_{\mathcal{W}}) + \rho(|u|)$, which satisfies Definition \ref{defn iISS Ly}. The implication  \Rmnum{2}$\Rightarrow$\Rmnum{3} is established in Section \ref{3.1}. The implication \Rmnum{3}$\Rightarrow$\Rmnum{4} follows  from the definitions: any iISS system is globally pre-stable and satisfies the asymptotic gain property, which are the conditions required for the nonuniformly iISS property. For the implication \Rmnum{4}$\Rightarrow$\Rmnum{5}, note first that the nonuniformly iISS property includes the asymptotic gain property by Definition \ref{defn noniiss}. Moreover, when the input vanishes ($u\equiv0$), the nonuniform iISS property implies 0-input pre-asymptotic stability. Global pre-stability guarantees that trajectories remain bounded under small inputs, while the pre-asymptotic stability ensures convergence. Consequently, for small initial conditions and small inputs, solutions remain in a neighborhood of $\mathcal{W}$ and converge asymptotically to it. This corresponds precisely to the definition of locally iISS. The implication \Rmnum{5} $\Rightarrow$\Rmnum{6} is immediate. Since the system is locally iISS, for $u \equiv 0$ and sufficiently small initial conditions, the state trajectory admits a class-$\mathcal{KLL}$ upper bound. This inequality directly implies 0-input pre-stability Lyapunov stability relative to $\mathcal{W}$. The asymptotic gain property holds by assumption. The implication  \Rmnum{6}$\Rightarrow$\Rmnum{1} is given in Section \ref{3.2}. From the implication \Rmnum{6}$\Rightarrow$\Rmnum{1} in Section \ref{3.2}, the system admits an iISS-LKF $V$ with an exponential decay. This exponential property ensures the existence of a function $\hat{\rho}\in\mathcal{K}$ such that for all $(\varphi, u)\in\mathcal{C}$: $D^+V(\varphi)\leq-\frac{v}{2}V(\varphi) + \hat{\rho}(|u|)$. By defining the dissipation rate functional $\psi(\varphi) := \frac{v}{2}V(\varphi)$, the system $\mathcal{H}_{\mathcal{M}}^{\Delta}$ is smoothly dissipative with respect to $\mathcal{W}$. Furthermore, we verify the 0-input detectability relative to the set $\{\xi\in\mathcal{M}^\Delta: \psi(\xi)=0\}$. Since $V$ is positive definite with respect to $\mathcal{W}$, the condition $\psi(\varphi)=0$ implies $V(\varphi)=0$, which implies $|\varphi(0,0)|_{\mathcal{W}}=0$. Consequently, any solution satisfying $\psi(\mathcal{A}_{[t,j]}^\Delta x) \equiv 0$ identically satisfies $|x(t,j)|_{\mathcal{W}}\equiv0$. Thus, the detectability condition is fulfilled, establishing the implication \Rmnum{6}$\Rightarrow$\Rmnum{7}. The implication  \Rmnum{8}$\Rightarrow$\Rmnum{1} is provided in Section \ref{3.3}. The implication \Rmnum{7}$\Rightarrow$\Rmnum{8} can be demonstrated as follows.

Let $V$ be the storage functional from item \Rmnum{7}, and let $u\equiv0$. 
The smooth dissipativity inequality reduces to $D^+V(\varphi)\le-\psi(\varphi)\le0$ for $\varphi\in\mathcal{C}$ and $V(\varphi^+)\le V(\varphi)$ for $\varphi\in\mathcal{D}$. 
Since $V$ is positive definite  and non-increasing along 0-input solutions, it serves as a Lyapunov-like functional, which implies that the set $\mathcal{W}$ is 0-input pre-stable. 
Consider a bounded complete solution pair $(x,0)$ to $\mathcal{H}_{\mathcal{M}}^{\Delta}$. Since the system has finite delay and satisfies Assumption~\ref{assume}, the boundedness of the solution in $\mathbb{R}^n$ implies the pre-compactness of the corresponding orbit in the state space $\mathcal{M}^\Delta$ via the Arzel\`{a}-Ascoli theorem \citep{liu2014hybrid}. 
Consequently, the invariance principle for hybrid systems \citep[Corollary~8.4]{HybridDynamicalSystemsModeling} remains applicable in the functional setting. 
Hence, there exists $r\geq0$ such that every complete solution of $\mathcal{H}_{\mathcal{M}}^{\Delta}$ converges to the largest weakly invariant set. 
Moreover, that weakly invariant set is a subset of the following set. 
\begin{equation}
	\label{variant set }
	\begin{split}
		N=\{\varphi: V(\varphi) = r\} \cap ( & \{\varphi \in \mathcal{C}: \psi(\varphi)=0\} \\
		& \cup \{\varphi \in \mathcal{D}: V(\varphi^+) = V(\varphi)\} ).
	\end{split}
\end{equation}
Since $\mathcal{W}$ is 0-input detectable relative to the set $N$ from Definition \ref{defn o-input detectable}, it follows that every 0-input complete solution $(x,0)$ whose corresponding $\mathcal{A}^\Delta x(t,j)$ belongs to $N$ for all $(t,j)$ with $t+j\ge T$ for some $T \ge 0$ necessarily converges to $\mathcal{W}$. From \eqref{eq:5}, we can see that if $r=0$, then $\{\varphi(0,0): V(\varphi)=0\}\subset\mathcal{W}$, and any invariant set contained in \eqref{variant set } is in $\mathcal{W}$. If $r>0$, then $\{\varphi(0,0): V(\varphi)=r\}\cap\mathcal{W}=\emptyset$, and thus any point  in \eqref{variant set } is not in $\mathcal{W}$. Therefore, $\mathcal{W}$ is 0-input pre-attractive.
\end{proof}

\subsection{Proof of \Rmnum{2}$\Rightarrow$\Rmnum{3}}  \label{3.1}
	Since the Lyapunov functional is non-increasing at jumps, we focus on the evolution of the Lyapunov functional in the flow set. For any $(t, j) \in \operatorname{dom} x$, we have 
\begin{equation} 
	\label{eq:integral_V}
	V(\mathcal{A}_{[t,j]}^\Delta x) - V(\varphi_0) \leq \sum_{n=0}^{j} \int_{t_n}^{t_{n+1}}  D^+V(\mathcal{A}_{[s,k(s)]}^\Delta x) \, ds,
\end{equation}
where $k(s)$ denotes the number of jumps that have occurred by time $s$. Substituting condition (2) of Definition \ref{defn iISS Ly} into \eqref{eq:integral_V}, we obtain:
$\label{eq:dissipation_int}
V(\mathcal{A}_{[t,j]}^\Delta x) - V(\varphi_0) \leq  -\sum_{n=0}^{j} \int_{t_n}^{t_{n+1}} \alpha_3(|x(s, n)|_{\mathcal{W}}) \, ds 
+ \sum_{n=0}^{j} \int_{t_n}^{t_{n+1}} \rho(|u(s, n)|) \, ds.
$
Since $\alpha_3$ is positive definite, the dissipation term is non-positive.
Removing it yields the following upper bound on the Lyapunov functional:
$ \label{eq:V_bound}
V(\mathcal{A}_{[t,j]}^\Delta x) \leq V(\varphi_0) + \sum_{n=0}^{j} \int_{t_n}^{t_{n+1}} \rho(|u(s, n)|) \, ds.
$
From  $\alpha_1(|x(t,j)|_{\mathcal{W}}) \leq V(\mathcal{A}_{[t,j]}^\Delta x)$ and $V(\varphi_0)\leq\alpha_2(\|\varphi_0\|_{\mathcal{W}})$, we have 
\begin{equation}
	\scalebox{0.885}{$
		|x(t,j)|_{\mathcal{W}} \leq 
		\alpha_1^{-1}\big( \alpha_2(\|\varphi_0\|_{\mathcal{W}})
		+ \sum_{n=0}^{j} \int_{t_n}^{t_{n+1}} \rho(|u(s,n)|)\, ds \big).
		$}
	\label{eq:BEBS}
\end{equation}
The inequality \eqref{eq:BEBS} implies that for any bounded initial state and any input with finite integral energy, the state trajectory remains bounded. This property is referred to as the \emph{Bounded Energy Bounded State (BEBS)} property.

Next if $u(t,j) \equiv 0$, then item (2) of Definition \ref{defn iISS Ly} reduces to the strict dissipation form: $D^+V(\varphi) \leq -\alpha_3(|\varphi(0,0)|_{\mathcal{W}})$, combining which with items (1) and (3) in Definition \ref{defn iISS Ly} yields a standard Lyapunov condition whose dissipation rate depends on the point-wise value of the state $|\varphi(0,0)|_{\mathcal{W}}$. Following the similar mechanism in the proof of \cite[Theorem 2]{LiuTeel}, for hybrid time-delay systems, such point-wise dissipation is sufficient to ensure global asymptotic stability. Specifically, there exists $\tilde{\beta} \in \mathcal{KLL}$ such that the zero-input solution satisfies:
\begin{equation} 
	\label{eq:0-GAS}
	|x(t,j)|_{\mathcal{W}} \leq \tilde{\beta}(\|\varphi_0\|_{\mathcal{W}}, t, j), \quad \forall (t,j) \in \operatorname{dom} x,
\end{equation}
which implies that the system is \emph{0-input globally asymptotically stable (0-GAS)}.

From the characterization of iISS, the iISS property is equivalent to the BEBS property and the 0-GAS property. Specifically, by combining \eqref{eq:BEBS} and \eqref{eq:0-GAS}, there exist $\beta\in\mathcal{KLL}$ and $\gamma\in\mathcal{K}$ such that $|x(t,j)|_{\mathcal{W}}\leq\max\{ 2\beta(\|\varphi_0\|_{\mathcal{W}}, t, j), 2\gamma(\sum_{n=0}^{j}\int_{t_n}^{t_{n+1}}\rho(|u(s, n)|)ds)\}$. Thus, the set $\mathcal{W}$ is iISS for $\mathcal{H}_{\mathcal{M}}^\Delta$.

\subsection{Proof of \Rmnum{6}$\Rightarrow$\Rmnum{1}}  \label{3.2}
	\begin{lemma}
	\label{cai}
	Consider the system \eqref{eq:hybrid}.
	Let $U\in\mathcal{U}\subset\mathbb{R}^m$ be a bounded closed set, and let $K_0,\ K_1 \subset \mathcal{M}^\Delta$ be bounded and closed sets such that $K_0+\epsilon\overline{\mathbb{B}}\subset K_1$ for some $\epsilon>0$. Assume that for each $\varphi\in K_1$, each solution $(x,u)\in \mathfrak{S}_u^\Delta(\varphi)$ with $u(t,j)\in U$ is bounded in $\mathbb{R}^n$, and furthermore, if the solution is complete then there exists $(T,J)\in\operatorname{dom}x$ such that $\mathcal{A}_{[T,J]}^\Delta x\in K_0$. Then the state trajectories starting from $K_1$ are bounded in $\mathcal{M}^\Delta$.
\end{lemma}

\begin{remark}
	Since $\mathcal{M}^\Delta$ is an infinite-dimensional space, closed bounded sets are not compact. Consequently, the classical finite covering arguments used in finite-dimensional proofs cannot be directly applied. However, the conclusion of Lemma \ref{cai} holds for bounded closed sets in this functional setting. This validity relies on the finite delay and the regularity conditions in Assumption \ref{assume}, which ensure the system map is locally bounded. Combined with the assumption that trajectories are bounded in $\mathbb{R}^n$  and eventually enter $K_0$, it follows that the entire functional trajectory is uniformly bounded in the norm of $\mathcal{M}^\Delta$.
	\hfill$\square$
\end{remark}

The proof of \Rmnum{6}$\Rightarrow$\Rmnum{1} has two steps. First, from Lemma \ref{cai} we establish that $\mathcal{H}_{\mathcal{M}}^{\Delta}$ is nonuniformly iISS. Next we demonstrate the existence of an iISS-LKF with an exponential decay.

Based on  \Rmnum{6}, the system is 0-input pre-stable and satisfies the asymptotic gain property. If $u \equiv 0$, then the asymptotic gain property implies global attractivity of the set $\mathcal{W}$, combining which the 0-input pre-stability implies that $\mathcal{W}$ is 0-input pre-asymptotically stable. To establish the implication \Rmnum{6}$\Rightarrow$\Rmnum{1}, we proceed by first proving that these properties imply global pre-stability. To this end, we employ a bounding function construction. Select a linear function $\hat{\gamma}\in\mathcal{K}_\infty$ such that $\hat{\gamma}(s)\geq \gamma(s)$ for all $s\in \mathbb{R}_{\geq0}$. Since $\hat{\gamma}$ is linear and dominates $\gamma$, we have: $\sum_{n=0}^{j}\int_{t_{n}}^{t_{n+1}}\gamma(|u(s)|)ds\leq\hat{\gamma}(\sum_{n=0}^{j}\int_{t_{n}}^{t_{n+1}}|u(s)|ds)$. Define the reachability bound function $\eta(r)$ as:
\begin{align} \label{eta_r}
	\eta(r)
	&:= \sup\{|x(t,j)|_{\mathcal{W}}
	:\;\mathcal{A}_{[0,0]}^\Delta x\in{\mathcal{M}^\Delta}_{2r}^{\mathcal{W}},(x,u)\in\mathfrak{S}_u^{\Delta}(\varphi),\notag \\
	& (t,j)\in\operatorname{dom}x,\,
	\sum_{n=0}^{j}\int_{t_{n}}^{t_{n+1}}|u(s)|\,ds\le\hat{\gamma}^{-1}(r)
	\}.
\end{align}
We show that $\eta(r)<\infty$ for all $r>0$. Fix any $r>0$. Consider any initial condition from the set $K_1 := {\mathcal{M}^\Delta}_{2r}^{\mathcal{W}}$ and any input $u$ with integral bounded by $\hat{\gamma}^{-1}(r)$. From the asymptotic gain property in item \Rmnum{6}, the state is bounded by the input gain. Since the input energy corresponds to a gain of at most $r$, we have $\limsup_{(t,j)\to \infty} |x(t,j)|_{\mathcal{W}} \le \gamma(\|u\|) \le \gamma(\hat{\gamma}^{-1}(r))\le r$. Let $K_0 := {\mathcal{M}^\Delta}_{3r/2}^{\mathcal{W}}$ be a bounded closed set. Since the asymptotic bound $r$ is strictly less than $3r/2$, there exists a time $(T,J)$ such that for all $(t,j) \succeq (T,J)$, the state satisfies $\mathcal{A}_{[t,j]}^\Delta x \in K_0$. Now we invoke Lemma \ref{cai} with the bounded input range $U=\mathcal{U}\cap\hat{\gamma}^{-1}(r)\overline{\mathbb{B}}$. The trajectory starting in the bounded set $K_1$ is locally bounded due to the system regularity, and eventually enters the bounded set $K_0$ as derived above. Thus, the conditions of Lemma \ref{cai} hold, and the entire reachable set is bounded. Consequently, $\eta(r)$ is bounded for all $r>0$.

Define $\hat\eta(s) := \sup_{0\le r\le s}\eta(r)$, which is non-decreasing and finite for all $s \ge 0$. 
Now, consider any initial memory arc $\varphi \in\mathcal{M}^\Delta$ and input $u \in \mathcal U$, and let $(x,u)\in\mathfrak S^\Delta_u(\varphi)$ be the corresponding solution. Set $r := \max\{\|\varphi\|_{\mathcal W}, \hat\gamma(\sum_{n=0}^j \int_{t_n}^{t_{n+1}}|u(s)|ds)\}$. Using the definition of $\eta(r)$ and the property $\sum_{n=0}^j \int_{t_n}^{t_{n+1}} \hat\gamma(|u(s)|)ds\le\hat\gamma(\sum_{n=0}^j \int_{t_n}^{t_{n+1}} |u(s)|\,ds)$, we obtain that for all $(t,j)\in \operatorname{dom} x$,  $|x(t,j)|_{\mathcal W} \le \eta(r)\le\hat\eta(r)$. Define $\tilde\alpha(s):=\hat{\eta}(s)$ and $\tilde\kappa(s) := \hat\eta(\hat\gamma(s))$. Thus, for all $(t,j)\in\operatorname{dom}x$, $|x(t,j)|_{\mathcal W}\le\max\{\tilde{\alpha}(\|\mathcal{A}_{[0,0]}^\Delta x\|_{\mathcal W}), \tilde\kappa(\sum_{n=0}^j \int_{t_n}^{t_{n+1}}|u(s)|ds)\}$, which establishes the global pre-stability of $\mathcal{H}_\mathcal{M}^\Delta$. Combining the above inequality with the asymptotic gain property, it follows from Definition~\ref{defn noniiss} that $\mathcal{H}_{\mathcal{M}}^\Delta$ is nonuniformly iISS.

Furthermore, we establish the existence of an iISS Lyapunov functional that decreases exponentially. First, we recall the explicit characterization of the asymptotic gain property in this context. The system \eqref{eq:hybrid} has the asymptotic gain property with $\gamma=\hat{\gamma}$ if and only if each solution pair $(x,u)\in\mathfrak{S}_u^{\Delta}(\varphi)$ is bounded and, if complete, satisfies:
$\limsup_{(t,j)\in \operatorname{dom} x,t+j\to \infty}|x(t,j)|_{\mathcal{W}}\leq
\hat\gamma(\sum_{n=0}^{j}\int_{t_{n}}^{t_{n+1}}\lvert u(s)\rvert \,ds )$. Let $\hat{\gamma}_1\in\mathcal{K}$ from  Definition \ref{defn asy } and let $\hat{\gamma}_2\in\mathcal{K}$ from  Definition \ref{defn gps }. Pick $\alpha\in\mathcal{K}_{\infty}$ such that for all $s\geq0$,
\begin{equation} 
	\label{6.1ineq}
	\scalebox{0.87}{\(\displaystyle
		\max\left\{\sum_{n=0}^{j}\int_{t_{n}}^{t_{n+1}}\hat{\gamma}_1(|\alpha(s)|)ds,	\sum_{n=0}^{j}\int_{t_{n}}^{t_{n+1}}\hat{\gamma}_2(|\alpha(s)|)ds
		\right\}       \leq \frac{s}{2}.
		\)}     
\end{equation}
Define the system  $\widehat{\mathcal{H}}_{\mathcal{M}}^{\Delta}:=\{ F,G,\widehat{C},\widehat{D}, \mathcal{M}^\Delta\}$ as follows,
\begin{align*}
	&\widehat{C}:= \{ \varphi : \exists u \in \mathcal{U} \cap \alpha(|\varphi(0,0)|_{\mathcal{W}})\overline{\mathbb{B}} \}, \\ &\widehat{D}:=\{ \varphi : \exists u \in \mathcal{U} \cap \alpha(|\varphi(0,0)|_{\mathcal{W}})\overline{\mathbb{B}} \ \text{such that} \ (\varphi,u)\in\mathcal{D} \}, \\
	&F(\varphi):= \{ f(\varphi,u) : u \in \mathcal{U} \cap \alpha(|\varphi(0,0)|_{\mathcal{W}})\overline{\mathbb{B}} \}, \quad \forall \varphi \in \mathcal{M}^\Delta, \\
	&G(\varphi):=\{ g(\varphi,u) : u \in \mathcal{U} \cap \alpha(|\varphi(0,0)|_{\mathcal{W}})\overline{\mathbb{B}} \}, \quad \forall \varphi \in \mathcal{M}^\Delta.
\end{align*}
$F(\varphi)=\overline{co}F(\varphi)$ for all $\varphi\in\widehat{C}$ and $\widehat{\mathcal{H}}_{\mathcal{M}}^{\Delta}$ satisfies Assumption \ref{assume}. The relation between $\widehat{\mathcal{H}}_{\mathcal{M}}^{\Delta}$ and $\mathcal{H}_{\mathcal{M}}^{\Delta}$ is given below.

\begin{lemma}
	\label{6.1lemma2}
	For each solution $x$ to $\widehat{\mathcal{H}}_{\mathcal{M}}^{\Delta}=(\mathcal{F},G,\widehat{\mathcal{C}},\widehat{\mathcal{D}},\mathcal{M}^\Delta)$, there exists a hybrid input $u$ such that $(x, u)$ is a solution pair to $\mathcal{H}_{\mathcal{M}}^{\Delta}$ and $|u(t,j)|\leq$ $\alpha(|x(t,j)|_\mathcal{W})$ for all $(t,j)\in\operatorname{dom}\;x$.
\end{lemma}

Let $\widehat{\mathfrak{S}}_u^{\Delta}$ denote the set of all maximal solutions to $\widehat{\mathcal{H}}_{\mathcal{M}}^{\Delta}$.
From the global pre-stability of $\mathcal{H}_{\mathcal{M}}^{\Delta}$, Lemma \ref{6.1lemma2} and  \eqref{6.1ineq}, for all $(t,j)\in \operatorname{dom} x$, each solution $(x,u)\in \widehat{\mathfrak{S}}_u^{\Delta}(\varphi)$ satisfies
\begin{align*}
	|x(t,j)|_{\mathcal{W}}&\leq\max \{\widehat{\alpha}(\|\mathcal{A}_{[0,0]}^\Delta x\|) , \\
	&0.5\underset{(s,i)\in \operatorname{dom} \;x,s+i\leq t+j}\sup|x(s,i)|_{\mathcal{W}} \}.
\end{align*}
If there exists $(t^*,j^*)$ such that $|x(t^*,j^*)|_{\mathcal{W}}>\hat{\alpha}(\|\mathcal{A}_{[0,0]}^\Delta x\|)$, then 
\begin{align*}
	|x(t^*,j^*)|_{\mathcal{W}}
	&\leq\max \{\hat{\alpha}(\|\mathcal{A}_{[0,0]}^\Delta x\|) , 0.5\sup_{\substack{(s,i)\in \operatorname{dom} x \\ s+i\leq t^*+j^*}} |x(s,i)|_{\mathcal{W}} \}\\
	&\leq0.5\sup_{\substack{(s,i)\in \operatorname{dom} x ,\;s+i\leq t^*+j^*}} |x(s,i)|_{\mathcal{W}} \\
	&\leq0.5|x(t^*,j^*)|_{\mathcal{W}}.
\end{align*}
This implies that $|x(t^*,j^*)|_{\mathcal{W}}=0$, which contradicts the assumption that $|x(t,j)|_{\mathcal{W}}>\hat{\alpha}(\|\mathcal{A}_{[0,0]}^\Delta x\|)$. Hence,  $|x(t,j)|_{\mathcal{W}}\leq \hat{\alpha}(\|\mathcal{A}_{[0,0]}^\Delta x\|)$ for all $(t,j)\in\operatorname{dom}x$. Therefore,  the set $\mathcal{W}$ is globally pre-stable for $\widehat{\mathcal{H}}_{\mathcal{M}}^{\Delta}$. From the asymptotic gain property and \eqref{6.1ineq}, any complete solution $(x,u)\in\widehat{\mathfrak{S}}_u^{\Delta}(\varphi)$ satisfies
\begin{align*}
	\underset{(t,j)\in \operatorname{dom} x,t+j\to \infty}{\lim\sup}|x(t,j)|_{\mathcal{W}}  \leq  
	0.5\underset{(t,j)\in \operatorname{dom} x,t+j\to \infty}{\lim\sup}|x(t,j)|_{\mathcal{W}},
\end{align*}
which implies that $\lim\sup_{(t,j)\in\operatorname{dom} x,t+j\to\infty} |x(t,j)|_{\mathcal{W}}=0$. Therefore, the set $\mathcal{W}$ is pre-attractive for $\widehat{\mathcal{H}}_{\mathcal{M}}^{\Delta}$.

Finally, we extend the converse Lyapunov theorem for pre-asymptotic stability \cite[Theorem~2]{6.1lemma} to the framework of hybrid systems with memory. 
Based on the functional extension framework introduced by \cite{may_Pepe}, the result in \cite{6.1lemma} can be generalized to the case of hybrid systems with memory. 
Hence, there exists a smooth Lyapunov functional $V:\mathcal{M}^\Delta \to \mathbb{R}_{\ge 0}$ for the system $\widehat{\mathcal{H}}_{\mathcal{M}}^\Delta$. 
Let $\rho:= \alpha^{-1}$ and $v:= 1-e^{-1}$, the conditions of Definition~\ref{defn exp ly} are satisfied. Hence, we establish the implication \Rmnum{6}$\Rightarrow$\Rmnum{1}, and construct an LKF for exponential iISS.

\subsection{Proof of \Rmnum{8}$\Rightarrow$\Rmnum{1}}   \label{3.3}
	\begin{theorem}
	\label{8.1}
	Consider the hybrid system $\mathcal{H}_{\mathcal{M}}^{\Delta}$ with memory and input $v\in\mathbb{R}^m$. Let $\mathcal{W}\subset \mathbb{R}^n$ be a closed set. The system $\mathcal{H}_{\mathcal{M}}^{\Delta}$ is 0-input pre-asymptotically stable if and only if there exist a smooth functional $V:\mathcal{M}^{\Delta} \to \mathbb{R}_{\ge 0}$,  $\alpha_1, \alpha_2, \zeta\in\mathcal{K}_\infty$,  $a\in(0,1]$, and a nonzero smooth function 
	$q: \mathbb{R}_{\ge 0} \to \mathbb{R}_{>0}$ with  $ q(s)\equiv 1$  for all  $s\in [0,1]$ such that
	\begin{flalign}
		&\alpha_1(|\varphi(0,0)|_{\mathcal{W}})\leq V(\varphi)\leq\alpha_2(\|\varphi\|_{\mathcal{W}}),
		\;\forall \varphi\in\mathcal{M}^\Delta;\\
		& \Braket{\nabla V(\varphi),f(\varphi,q(|\varphi(0,0)|_\mathcal{W})v)}\leq -aV(\varphi),  
		\;\label{lemma8.1.1}\notag\\
		&\quad
		\forall (\varphi,q(|\varphi(0,0)|_\mathcal{W})v)\in\mathcal{C},\; |\varphi(0,0)|_\mathcal{W}\geq \zeta(|v|)	\\ 
		&V(g(\varphi,q(|\varphi(0,0)|_\mathcal{W})v))-V(\varphi)		\leq0,\notag\\
		&\quad\qquad\quad\qquad\quad
		\forall (\varphi,q(|\varphi(0,0)|_\mathcal{W})v)\in\mathcal{D}.   \label{lemma8.1.2}
	\end{flalign}	 
	Here, $v\in\mathbb{R}^m$ is an external input, and $q(|\varphi(0,0)|_\mathcal{W}) v$ is a virtual input, which can be reduced to the 0-input case if $v=0$. 
\end{theorem}

\begin{proof}
	The sufficiency holds by setting $v = 0$. Next we prove the necessity. By extending the converse Lyapunov theorem results in \cite{6.1lemma}, there exist a smooth functional $V$, $\alpha_1, \alpha_2 \in \mathcal{K}_{\infty}$, and a decay rate $a^* > 0$ such that:
	\begin{align*}
		\alpha_1(|\varphi(0,0)|_{\mathcal{W}})\leq V(\varphi)&\leq\alpha_2(\|\varphi\|_{\mathcal{W}}),\quad&&\forall \varphi\in\mathcal{M}^\Delta;\\
		\Braket{\nabla V(\varphi),f(\varphi,0)} &\leq -a^* V(\varphi), \quad &&\forall \varphi\in\mathcal{C}.\\
		V(g(\varphi,0))-V(\varphi)		&\leq0,	&&\forall \varphi\in\mathcal{D}.
	\end{align*}
	Define the function $\delta:\mathbb{R}_{\geq0}\times \mathbb{R}_{\geq0}\to \mathbb{R}$ by
	\begin{equation*}
		\begin{split}
			\delta(s,r) := \sup_{\substack{|\varphi(0,0)|_\mathcal{W}=s ,|\mu|=r}} \max \{ 
			& \langle \nabla V(\varphi),f(\varphi,\mu) \rangle +aV(\varphi), \\
			& V(g(\varphi,\mu)) - V(\varphi)
			\}.
		\end{split}
	\end{equation*}
	We aim to establish \eqref{lemma8.1.1} with $a \in (0, a^*)$. Consider the case $r=0$. Since $a < a^*$, we have
	$\langle\nabla V(\varphi),f(\varphi,0) \rangle +aV(\varphi) \leq -a^* V(\varphi) + a V(\varphi) = -(a^*-a)V(\varphi)$. Since $V(\varphi) \ge \alpha_1(s)$, for any $s > 0$, the term $-(a^*-a)V(\varphi)$ is strictly negative. Thus, $\delta(s,0) < 0$ for all $s > 0$ as $\mathcal{H}_{\mathcal{M}}^{\Delta}$ is 0-input pre-asymptotically stable. From \cite{8.1lemma}, there exist $\zeta\in\mathcal{K}_\infty$ and a continuous function $q\colon\mathbb{R}_{\geq0}\to\mathbb{R}_{\geq0}$ such that
	\begin{enumerate}
		\item $q(s) \neq0$ for all $s\geq0$, and $q(s) \equiv1$ for all $s\in$ [0,1];
		\item $\delta(s,p)<0$ for each pair $(s,r)\in\mathbb{R}_{\geq0}\times\mathbb{R}_{\geq0}$ with $\zeta(r)<s$ and each $p\leq q(s)r$. 
	\end{enumerate}
	Since $\delta(s,0) < 0$ for $s > 0$, there exist $\zeta \in \mathcal{K}_\infty$ and a smooth function $q: \mathbb{R}_{\geq 0} \to \mathbb{R}_{> 0}$ with $q(s) \equiv 1$ for $s \in [0,1]$ such that $\delta(s, p) < 0$ holds for all  $s > \zeta(p/q(s))$. This property ensures that the time derivative and jump difference of the functional are negative definite provided the state norm dominates the scaled input. Now, we verify the theorem conditions. Let $s = |\varphi(0,0)|_\mathcal{W}$ and $r = |v|$. Suppose the state satisfies the bound $|\varphi(0,0)|_\mathcal{W} \ge \zeta(|v|)$. Let the scaled input magnitude be $p = q(s)r$. The condition $s \ge \zeta(r)$ implies $s > \zeta(p/q(s))$. Then we have $\delta(s, p) < 0$. Recalling the definition of $\delta(s,p)$, $\delta(s,p) < 0$ implies $\max\{ \langle \nabla V(\varphi), f(\varphi, q(s)v) \rangle + a V(\varphi), V(g(\varphi, q(s)v))-V(\varphi)\}<0$. Hence, \eqref{lemma8.1.1} and \eqref{lemma8.1.2} hold.
\end{proof}

\begin{proposition}    
	\label{end8.1}
	The system 	$\mathcal{H}_{\mathcal{M}}^{\Delta}$ is 0-input pre-asymptotically stable if and only if there exist a smooth  functional $\mathcal{P}:\mathcal{M}^\Delta \to \mathbb{R}_{>0}$, $\bar{\alpha}_1, \;\bar{\alpha}_2,\;\lambda\in \mathcal{K}$ and $c^*>0$ such that
	\begin{flalign}  
		&\bar{\alpha}_1(|\varphi(0,0)|_{\mathcal{W}})\leq \mathcal{P}(\varphi)\leq\bar{\alpha}_2(\|\varphi\|_{\mathcal{W}}),\quad
		&&\!\!\!\!\!\!
		\forall \varphi\in\mathcal{M}^\Delta;\\
		&\langle\nabla \mathcal{P}(\varphi), f(\varphi,u)\rangle \leq -c^*\mathcal{P}(\varphi) + \lambda(|u|) ,
		&&\!\!\!\!\!\!
		 \forall (\varphi,u) \in \mathcal{C};\label{end8.1.1} \\
		&\mathcal{P}(g(\varphi,u)) - \mathcal{P}(\varphi) \leq 0 ,
		&&\!\!\!\!\!\!\!\!
		\forall (\varphi,u) \in \mathcal{D}.\label{end8.1.2} 
	\end{flalign}
	
\end{proposition}

\begin{proof}
	If the smooth  functional $\mathcal{P}$ exists, then for $u \equiv 0$, the condition \eqref{end8.1.1} ensures the strict decrease of $\mathcal{P}$ along flows, while the condition \eqref{end8.1.2} guarantees that $\mathcal{P}$ does not increase at jumps. 	Hence, the functional $\mathcal{P}$ serves as an LKF to ensure the 0-input pre-asymptotic stability of $\mathcal{H}_{\mathcal{M}}^{\Delta}$. 
	
	We now establish the necessity part. Let $\mathcal{H}_{\mathcal{M}}^{\Delta}$ be 0-input pre-asymptotically stable. From Theorem \ref{8.1}, there exists a Lyapunov functional $V:\mathcal{M}^\Delta \to \mathbb{R}_{>0}$, and there exists $\rho\in\mathcal{K}_{\infty}$ such that \eqref{lemma8.1.1} and \eqref{lemma8.1.2} imply
	\begin{align}
		\Braket{\nabla V(\varphi),f(\varphi,q(|\varphi(0,0)|_\mathcal{W})v)}&\leq -aV(\varphi)+\rho(|v|),\notag\\
		&\!\!\!\!\!\!\!\!\!\!\!\!\!\!\!\!\!\!\!\!\!\!\!\!\!\!\!\!\!\!\!\!\! \forall (\varphi,q(|\varphi(0,0)|_\mathcal{W})v)\in\mathcal{C};\\ 
		V(g(\varphi,q(|\varphi(0,0)|_\mathcal{W})v))-V(\varphi)		&\leq0,\notag\\
		&\!\!\!\!\!\!\!\!\!\!\!\!\!\!\!\!\!\!\!\!\!\!\!\!\!\!\!\!\!\!\!\!\! \forall (\varphi,q(|\varphi(0,0)|_\mathcal{W})v)\in\mathcal{D}.         
	\end{align}
	From Corollary IV.5 in \cite{8.1proof}, there exists $\lambda\in \mathcal{K}$ such that $\rho(sr)\leq \lambda(s)\lambda(r)$ for all $(s,r)\in\mathbb{R}_{\geq0}\times\mathbb{R}_{\geq0}$. Hence,
	\begin{align}
		\Braket{\nabla V(\varphi),f(\varphi,u)}&\leq -aV(\varphi)+\lambda(1/q(|\varphi(0,0)|_\mathcal{W}))\lambda(|u|),\notag\\
		&\qquad\qquad\qquad\quad\forall (\varphi,u)\in\mathcal{C},\\
		V(g(\varphi,u))-V(\varphi)		&\leq0, \qquad\qquad\;\;\; \forall (\varphi,u)\in\mathcal{D},
	\end{align}
	where $u:=q(|\varphi(0,0)|_\mathcal{W})v$. 
	Define $\pi(r):=\int_{0}^{r} \frac{ds}{\tilde{\omega}+\theta(s)}$, where $\tilde{\omega}>0$, and $\theta\in\mathcal{K}$. Let $\mathcal{P}=\pi\circ V $. Since $\mathcal{P}(\varphi)=\int_0^{V(\varphi)}\frac{ds}{\tilde{\omega}+\theta(s)}\leq\int_0^{V(\varphi)}\frac{ds}{\tilde{\omega}}=\frac{V(\varphi)}{\tilde{\omega}}$, for all $\varphi\in\mathcal{C}$, we have $\frac{aV(\varphi)}{(\tilde{\omega}+\theta(V(\varphi)))\mathcal{P}(\varphi)}\geq\frac{aV(\varphi)}{(\tilde{\omega}+\theta(V(\varphi)))\cdot(V(\varphi)/\tilde{\omega})}=\frac{a\tilde{\omega}}{\tilde{\omega}+\theta(V(\varphi))}$.
	Then, 
	\begin{align*}
		\langle\nabla \mathcal{P}(\varphi),f(\varphi,u)\rangle&\leq  \frac{\Braket{\nabla V(\varphi),f(\varphi,u)}}{\tilde{\omega}+\theta(V(\varphi))}\;\\
		&\leq -\frac{aV(\varphi)}{\tilde{\omega}+\theta(V(\varphi))}   +\frac{\lambda(\frac{1}{q(|\varphi(0,0)|_\mathcal{W})})\lambda(|u|)}{\tilde{\omega}+\theta(V(\varphi))}\\
		&\leq -\frac{a\tilde{\omega}}{\tilde{\omega}+\theta(V(\varphi))} \mathcal{P}(\varphi)\\ &\;\; +\frac{\lambda(\frac{1}{q(|\varphi(0,0)|_\mathcal{W})})\lambda(|u|)}{\tilde{\omega}+\theta(\alpha_1(|\varphi(0,0)|_\mathcal{W}))}\\
		&\leq -\frac{a\tilde{\omega}}{\tilde{\omega}+\theta(\alpha_2(\|\varphi\|_{\mathcal{W}}))} \mathcal{P}(\varphi) \\ &\;\; +\frac{\lambda(\frac{1}{q(|\varphi(0,0)|_\mathcal{W})})\lambda(|u|)}{\tilde{\omega}+\theta(\alpha_1(|\varphi(0,0)|_\mathcal{W}))},\\
		\mathcal{P}(g(\varphi,u))-\mathcal{P}(\varphi)&\leq    \frac{	V(g(\varphi,u))-V(\varphi)}{\tilde{\omega}+\theta(V(\varphi))}
		\leq 0.
	\end{align*}
	In order to  ensure the coefficient of $\lambda(|u|)$ is bounded by 1, we can construct $\theta\in \mathcal{K}$ such that for all $s\in\mathbb{R}_{\geq0}$:
	\begin{equation}  
		\label{endproof8.1prop}
		\tilde{\omega}+\theta(\alpha_1(s))\geq \lambda(1/q(s)) ,
	\end{equation}
	Note that $q$ is smooth and $\tilde{\omega}>0$ can be defined appropriately. Define $c(\varphi)=\frac{a\tilde{\omega}}{\tilde{\omega}+\theta(\alpha_2(\|\varphi\|_{\mathcal{W}}))}$, and thus $c(\varphi)\in (0, a)$. There exists $c^* \in (0, a)$ such that $c(\varphi) \geq c^*$. As a result, 
	\begin{align*}
		\langle\nabla \mathcal{P}(\varphi),f(\varphi,u)\rangle
		&\leq-c^*\mathcal{P}(\varphi)+\lambda(|u|),\;&&\forall(\varphi,u)\in\mathcal{C},\\
		\mathcal{P}(g(\varphi,u))-\mathcal{P}(\varphi)&\leq0,\;\qquad&&\forall(\varphi,u)\in\mathcal{D}.
	\end{align*}
	This proves the necessity.
\end{proof}

The implication \Rmnum{8}$\Rightarrow$\Rmnum{1} follows  from the established results.
 Assumption \Rmnum{8} guarantees that the system is 0-input pre-asymptotically stable. From Proposition \ref{end8.1}, this property implies the existence of a smooth functional $\mathcal{P}: \mathcal{M}^\Delta \to \mathbb{R}_{>0}$ satisfying the dissipation inequalities \eqref{end8.1.1} and \eqref{end8.1.2} with a constant decay rate $c^*$. Identifying $V(\varphi) = \mathcal{P}(\varphi)$, $\rho(|u|) = \lambda(|u|)$, and $v = c^*$, this functional $\mathcal{P}$ precisely fulfills the requirements of an exponential iISS-LKF as per Definition \ref{defn exp ly}. Thus, the existence of such a functional is proved.

\section{EXPERIMENTS}
	Consider a quadcopter with the translational dynamics: $\dot{x}=\bar{A}x+\bar{B} u_{\mathrm{total}}$, where the state $x=[p^\top, v^\top]^\top \in\mathbb{R}^6$ consists of the position and velocity; see \cite{example1} for more details. $\bar{A}=[0_3, I_3; 0_3, -\tfrac{d_c}{m}I_3]$ and $\bar{B}=[0_3; \tfrac{1}{m}I_3]$, where $m>0$ and $d_c\in\mathbb{R}$ are respectively the mass and drag coefficient. The quadcopter is controller remotely via a wireless communication network, and thus both the communication delay $r>0$ and actuator constraints are involved. Hence, a swtiching controller is designed as
$u_{\mathrm{total}} = \mathbf{sat}(K_l x(t-r))+u$, where $K_l\in\mathbb{R}^6$ is the feedback gain for different mode $l\in\mathcal{Q}\subset\mathbb{N}$, $u\in\mathbb{R}^3$ is the external input, and $\mathbf{sat}(\cdot)$ is a component-wise saturation function with a maximal limit $u_{\max}$. Here let $\mathcal{Q}=\{1,2\}$.
As a result, the closed-loop system is modeled as a hybrid system with memory $\mathcal{H}_\mathcal{M}^\Delta=(\mathcal{F},\mathcal{G},\mathcal{C},\mathcal{D})$ with  
\begin{align*}
	\scalebox{0.9}{$\displaystyle
		\mathcal{F}(\varphi,u):=\begin{bmatrix}\dot{x} = \bar{A} \psi(0,0) + \bar{B}\mathbf{sat}(K_\ell \psi(-r,k(-r))) + \bar{B} u(0,0)\\0\\1\end{bmatrix}
		$}
\end{align*}
with $\varphi=(\psi,\ell,\tau)$ and $\mathcal{C}=\{\varphi\in\mathcal{M}^\Delta: \tau(0,0)\in[0,\delta],\ell(0,0)\in\mathcal{Q}\}$, and 
\begin{align*}
	\scalebox{0.9}{$\displaystyle
	\mathcal{G}(\varphi,u)\in\begin{bmatrix}\bar{D}\psi(0,0)\\\mathcal{Q}\\0\end{bmatrix}
		$}
\end{align*}
with $\mathcal{D}=\{\varphi\in\mathcal{M}^{\Delta}: \tau(0,0)=\delta,\ell(0,0)\in\mathcal{Q}\}$, where $\psi$ represents the continuous state, $k(s)=\max\{k:(s,k)\in\operatorname{dom}\varphi\}$, $\ell: \mathbb{R}\rightarrow\mathcal{Q}$ is the switching signal, $\tau$ is a timer to track the switching intervals with the bound $\delta>0$, and $\bar{D}$ is a matrix with appropriate dimensions. 

Let $\mathcal{W}:=\{0\}\times\mathcal{Q}\times[0,\delta]$. Since $\varphi=(\psi,\ell,\tau)$ and $\tau\in[0,\delta]$, we have $|\varphi(0,0)|_\mathcal{W}=|\psi(0,0)|$ and $\|\varphi\|_\mathcal{W}=\|\psi\|$ for all $\varphi\in\mathcal{C}\cup \mathcal{D}\subset\mathcal{M}^\Delta $.
 For all $\varphi\in\mathcal{M}^{\Delta}$, define a functional  $V(\varphi):=\sigma_{\ell}|\psi(0,0)|^2+\mu_{\ell}\int_{-r}^0e^{\eta s}|\psi(s,k(s))|^2ds$ with $\sigma_{\ell},\mu_{\ell}>0,\eta\geq0$. Thus, item (1) of Definition 5 holds with $\alpha_1(s)=\min_{\ell\in\mathcal{Q}}\{\sigma_{\ell}\}s^2$, and $\alpha_2(s)=\max_{\ell\in\mathcal{Q}}\{\sigma_{\ell}\:+\mu_\ell r\}s^2$.
For all $\varphi\in\mathcal{C}$, $D^{+}V(\varphi)= D^+ ( \sigma_{\ell} |\psi(0,0)|^2 ) + D^+ ( \mu_{\ell} \int_{-r}^0 e^{\eta s}|\psi(s,k(s))|^2 ds )$. In particular, $D^+(\sigma_{\ell}|x(t)|^2)= 2\sigma_{\ell}x(t)^\top\dot{x}(t) 
=2\sigma_{\ell} x(t)^\top\bar{A} x(t)+2\sigma_{\ell}x(t)^\top \bar{B} \mathbf{sat}(K_\ell x(t-r)) + 2\sigma_{\ell} x(t)^\top \bar{B}u$. The derivative of the integral term is bounded as: $D^+(\mu_{\ell}\int_{t-r}^t e^{\eta(s-t)}|x(s)|^2 ds)=\mu_{\ell}|x(t)|^2-\mu_{\ell} e^{-\eta r}|x(t-r)|^2-\eta\mu_{\ell} \int_{t-r}^te^{\eta(s-t)} |x(s)|^2 ds \leq \mu_{\ell} |\psi(0,0)|^2-\mu_{\ell}e^{-\eta r}|\psi(-r, k(-r))|^2$. For notational brevity, let $\psi(-r)$ denote $\psi(-r, k(-r))$.
We obtain
\begin{align*}
	D^+V(\varphi) \leq& \psi(0,0)^\top (\sigma_{\ell}(\bar{A}+\bar{A}^\top)) \psi(0,0) + \mu_{\ell} |\psi(0,0)|^2 
	\\&- \mu_{\ell} e^{-\eta r} |\psi(-r)|^2+ 2\sigma_{\ell}\psi(0,0)^\top\bar{B}\mathbf{sat}(K_\ell \psi(-r)) 
	\\&+ 2\sigma_{\ell} \psi(0,0)^\top \bar{B} u.
\end{align*}
From the saturation function, $|\mathbf{sat}(K_\ell \psi(-r))| \leq |K_\ell \psi(-r)| \leq |K_\ell||\psi(-r)|$. 
Let $\Lambda:=2\psi(0,0)^\top (\sigma_{\ell}\bar{B}\mathbf{sat}(K_\ell \psi(-r)))$. 
From Young's inequality, there exists $\epsilon_{1} > 0$ such that 
\begin{align*}
	\begin{aligned}
		\Lambda &\leq \epsilon_{1} |\psi(0,0)|^2 + \frac{1}{\epsilon_{1}} |\sigma_{\ell} \bar{B}\mathbf{sat}(K_\ell \psi(-r))|^2 \\
		&\leq \epsilon_{1} |\psi(0,0)|^2 + ( \frac{\sigma_{\ell}^2 |\bar{B}|^2 |K_\ell|^2}{\epsilon_{1}})|\psi(-r)|^2, 
	\end{aligned}
\end{align*}
and there exists $\epsilon_{2}>0$ such that 
\begin{align*}
	\begin{aligned}
		2\psi(0,0)^\top(\sigma_{\ell}\bar{B}u)&\leq\epsilon_{2}|\psi(0,0)|^2+\frac{1}{\epsilon_{2}}|\sigma_{\ell}\bar{B}u|^2 \\
		&\leq\epsilon_{2}|\psi(0,0)|^2+\frac{\sigma_{\ell}^2|\bar{B}|^2}{\epsilon_{2}}|u|^2.
	\end{aligned}
\end{align*}
Hence, for all $\varphi\in\mathcal{C}$, $D^+V(\varphi) \leq ( 2\sigma_{\ell} \lambda_{\max}(\frac{\bar{A}+\bar{A}^\top}{2}) + \mu_{\ell} + \epsilon_{1} + \epsilon_{2} ) |\psi(0,0)|^2 + ( \frac{\sigma_{\ell}^2 |\bar{B}|^2 |K_\ell|^2}{\epsilon_{1}} - \mu_{\ell} e^{-\eta r}) |\psi(-r)|^2 +\frac{\sigma_{\ell}^2 |\bar{B}|^2}{\epsilon_{2}}|u|^2$. For all $\varphi\in\mathcal{D}$, $V(\mathcal{G}(\varphi,u))-V(\varphi)\leq\sigma_{\ell}(\lambda_{\max}^2(\bar{D})-1)|\psi(0,0)|^2$.

\begin{figure}[!t]
	\centering
	\begin{subfigure}[c]{0.48\columnwidth}
		\centering
		\includegraphics[width=\linewidth,keepaspectratio]{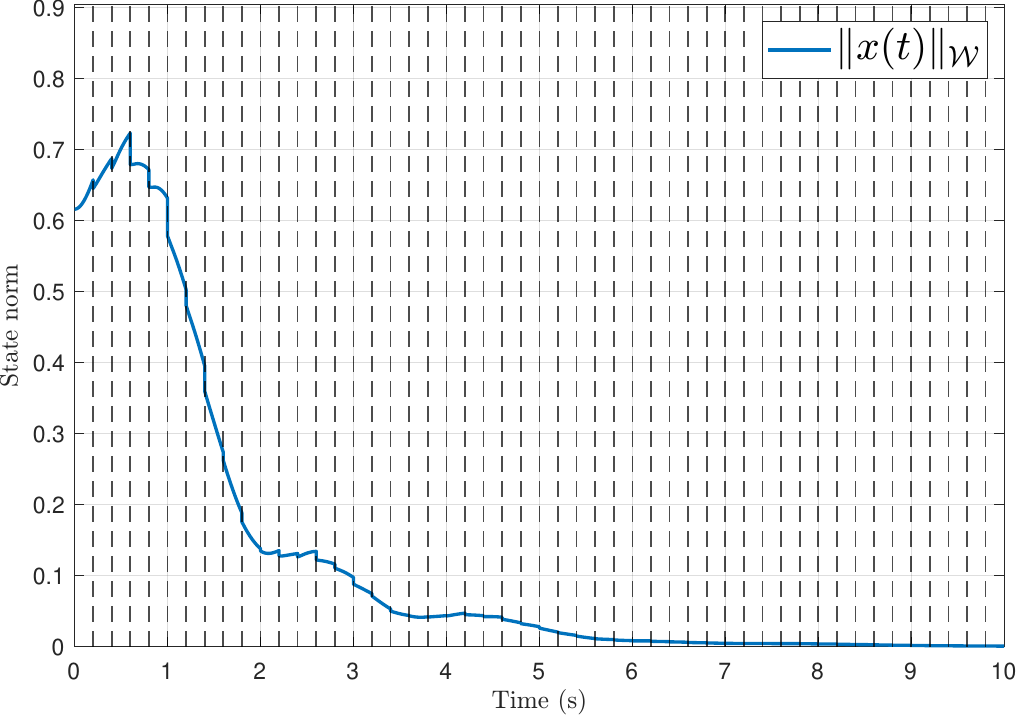}
		\caption{The time evolution of $\|x(t)\|_\mathcal{W}$ under $u_1$.}
		\label{fig1}
	\end{subfigure}\hfill
	\begin{subfigure}[c]{0.48\columnwidth}
		\centering
		\includegraphics[width=\linewidth,keepaspectratio]{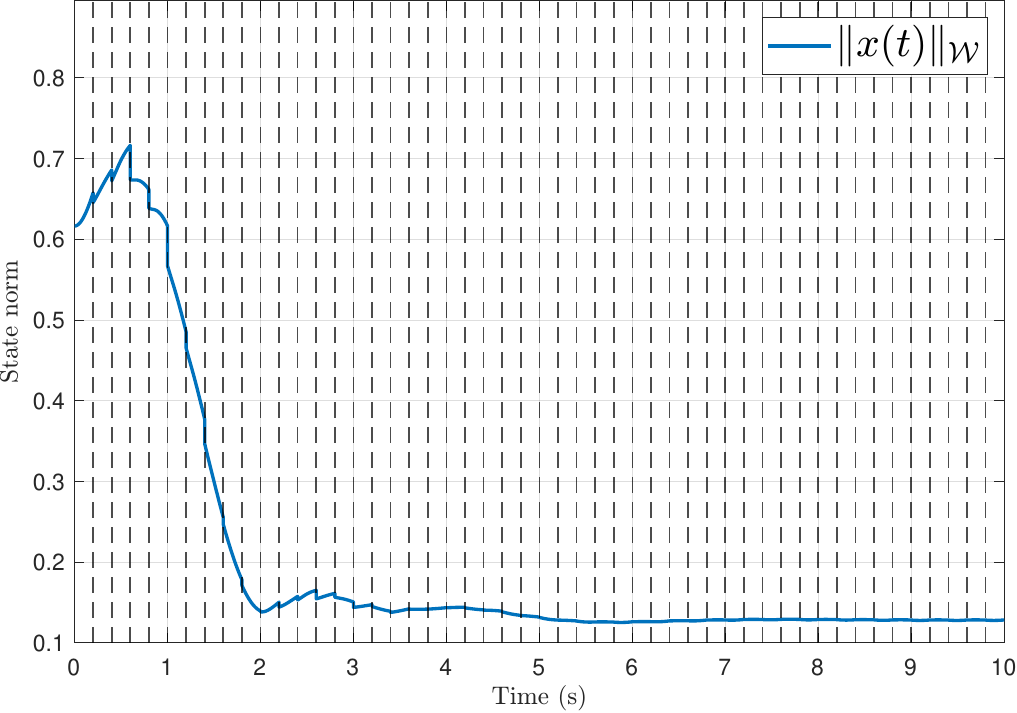}
		\caption{The time evolution of $\|x(t)\|_\mathcal{W}$ under $u_2$.}
		\label{fig2}
	\end{subfigure}
	\\[4pt]
	
	\begin{subfigure}[c]{0.48\columnwidth}
		\centering
		\includegraphics[width=\linewidth,keepaspectratio]{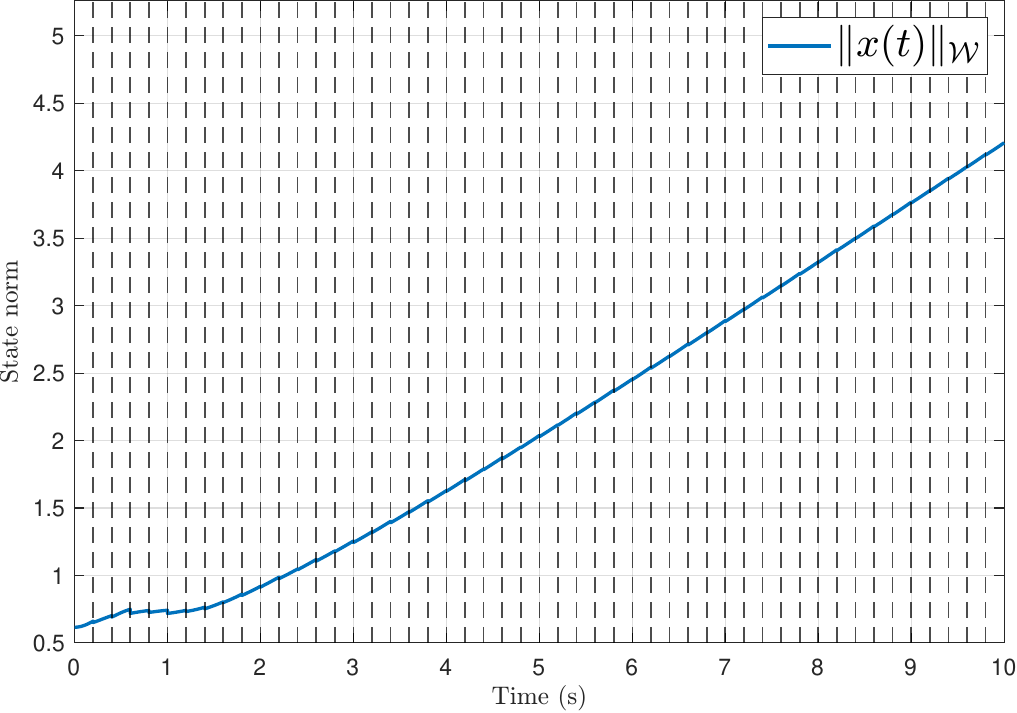}
		\caption{The time evolution of $\|x(t)\|_\mathcal{W}$ under $u_3$.}
		\label{fig3}
	\end{subfigure}\hfill
	\begin{subfigure}[c]{0.5\columnwidth}
		\centering
		\includegraphics[width=\linewidth,keepaspectratio]{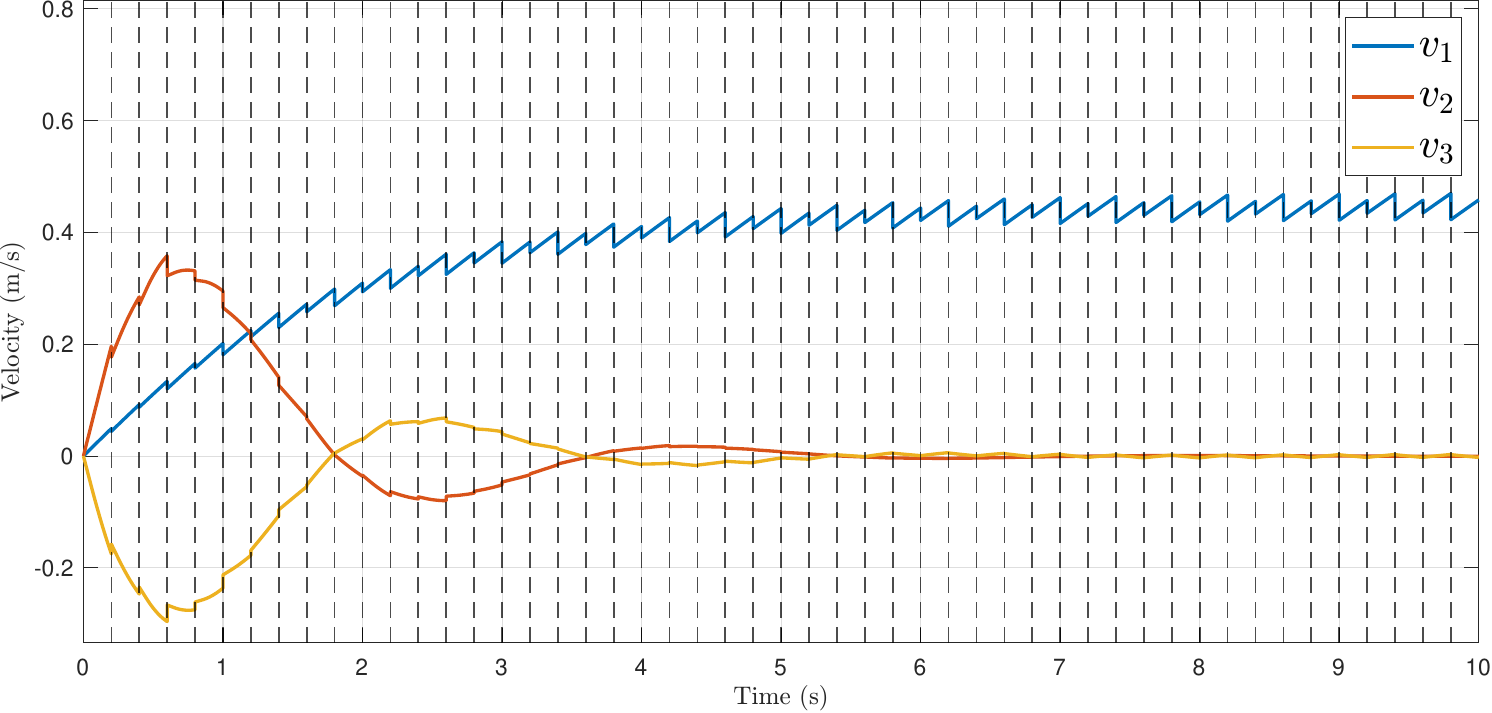}
		\caption{The time evolution of velocity trajectories ($v_1, v_2, v_3$) under $u_3$.}
		\label{fig4}
	\end{subfigure}	
	\caption{Time evolution of the state norm $\|x(t)\|_\mathcal{W}$ and velocity trajectories under different control inputs.}
	\label{fig:all}
\end{figure}

The simulation parameters are set as $m = 1.2$, $d_c = 0.2$, $u_{\max} = 1.0$, $r = 0.05\,\mathrm{s}$, and $\delta = 0.2\,\mathrm{s}$. The feedback gain matrix is structured as $K_\ell = [-k_{p\ell}I_3, -k_{d\ell}I_3]$. The control gains are selected as $k_{p1} = 4.8, k_{d1} = 1.5$ and $k_{p2} = 3.6, k_{d2} = 1.3$. Three input scenarios are tested: a vanishing input $u_1=[0.5, 0, -0.2]^\top e^{-0.5 t}$ and two constant inputs $u_2=[0.5, 0, -0.2]^\top, u_3=[1.5, 0, -0.2]^\top$.

From Fig.~\ref{fig:all}, the system stability is guaranteed under different external inputs. For the finite-energy input $u_1$ (Fig.~\ref{fig1}), $\|x(t)\|_\mathcal{W}$ converges to zero, consistent with the 0-GAS requirement of iISS. In the case of the small persistent input $u_2$ (Fig.~\ref{fig2}), the actuator remains unsaturated, and the state converges to a bounded neighborhood (asymptotic gain property). Conversely, the large input $u_3$ saturates the first control channel (Fig.~\ref{fig4}) since the maximum control effort ($1.0$) cannot cancel the disturbance ($1.5$). Consequently, $v_1$ converges to a positive constant, causing the position and the state norm to grow linearly to infinity (Fig.~\ref{fig3}). This unbounded trajectory under a bounded input confirms that the system is not ISS,  yet it remains iISS.

\section{CONCLUSION}

This paper established a unified iISS framework for hybrid systems with memory. By leveraging the Krasovskii approach, we proved the equivalence among Lyapunov criteria, dissipativity, and storage functionals. These findings effectively resolve the analytical challenges posed by the coupling of memory-dependent dynamics and impulsive effects. Future research will target networked and event-triggered control to address resource constraints, with experimental validation on physical testbeds.

\bibliographystyle{ieeetr}
\bibliography{cccconf} 
%
%
%
%

\end{document}